\documentclass{article} 

\usepackage{hyperref}
\usepackage{url}
\usepackage{natbib}
\usepackage{booktabs}
\usepackage{enumitem}
\usepackage{booktabs}
\usepackage{threeparttable}
\usepackage{tikz}
\usepackage{multirow}
\usepackage{amsmath,amssymb}
\usepackage{caption}
\usepackage{subcaption}
\usepackage{appendix}

\usetikzlibrary{arrows.meta, positioning, shapes.geometric}
\usepackage{amsthm}
\newtheorem{assumption}{Assumption}
\newtheorem{theorem}{Theorem}
\newtheorem{lemma}{Lemma}
\newtheorem{proposition}{Proposition}
\newtheorem{corollary}{Corollary}
\newtheorem{remark}{Remark}
\newtheorem{example}{Example}
\newtheorem{definition}{Definition}
\usepackage{float} 
\usepackage[margin=1in]{geometry}

\title{Regret Equals Covariance: A Characterization for Stochastic Optimization}

\author{Aldridge, I. \thanks{https://www.irenealdridge.com/} \\
Cornell Tech\\
New York, NY, 10044 USA \\
\texttt{irene.aldridge@gmail.com} 
}

\begin{document}

\maketitle

\begin{abstract}
Quantifying regret typically involves repeatedly re-solving a stochastic optimization problem across simulated scenarios. Such an approach is costly: the cost scales with the number of scenarios and the cost of each solve. We propose that regret can instead be measured using the joint moments of the cost vector $c$ and the optimal decision $\pi^*(c)$. Working with the optimal concave and Lipschitz under only a feasible region value function $v(c) = \min_{z\in Z} c^\top z$, we show that regret always decomposes as $\mathrm{Regret} = -\mathrm{Cov}(c,\pi^*(c)) - R(c)$ for a residual $R(c)$. We further show that $R(c)=0$  when $\pi^*$ is affine in $c$, which holds for unconstrained and equality-constrained (e.g., budget-constrained Markowitz) quadratic programs and many, but not all, linear programs. For all cases, we give a distribution-free bound $\mathrm{Regret} \le D\sqrt{\mathrm{tr}(\Sigma_c)}$ that depends only on the radius of the feasible region and the cost covariance. We also derive concentration and asymptotic-normality results for the empirical regret, the latter under an explicit and checkable non-degeneracy condition. 
\end{abstract}


\section{Introduction}
\label{sec:intro}
 
A basic task in decision-making under uncertainty is to quantify how
much is lost by committing to a single decision before true values are revealed, rather than tailoring the decision to each realization as
it becomes known. This quantity can be thought of as the expected value of perfect information. In the language of machine learning and predict-then-optimize modeling, this quantity is known as regret. Regret is a natural diagnostic wherever a stochastic optimization problem $\min_{z\in Z}\, c^\top z$ is solved using only a point estimate of an uncertain cost vector $c$: in portfolio selection
\citep{markowitz1952portfolio}, in the broader family of contextual and
data-driven optimization problems surveyed by \citet{sadana2023survey}, and as the target loss in the smart
predict-then-optimize framework of \citet{elmachtoub2022smart}. At the time this paper was written, the standard way to estimate regret was to use Sample Average Approximation: draw $B$ cost scenarios, re-solve the optimization problem for each, and average \citep{kleywegt2002sample,king2012modeling}. This method is reliable but expensive, as it requires as many solves as there are scenarios. A
natural question arises: can we compute regret directly from the moments of $c$ and the corresponding optimal decisions, without
resolving the optimization problem scenario by scenario?
 
The answer depends on a single structural property, and 
distinguishing when that property holds from when it does not is the
substance of this paper. Write $\pi^*(c) \in \arg\min_{z\in Z} c^\top z$ for the decision that would be optimal under a known realization of $c$, and $\mu = \mathbb E[c]$. Adding and subtracting $\mathbb E[c]^\top \mathbb E[\pi^*(c)]$ in the definition of regret produces an algebraic identity, $\mathrm{Regret} = -\mathrm{Cov}(c,\pi^*(c)) - R(c)$, in which all the difficulty is concentrated in the residual
term $R(c) = \mathbb E[c]^\top(\mathbb E[\pi^*(c)] - \pi^*(\mu))$. This
identity by itself is content-free: it holds by construction for any
decision rule. We show that $R(c)=0$ when $\pi^*$ is affine in $c$, a property of the \emph{decision rule}, not of the cost distribution's smoothness or compactness. We also show that $R(c)=0$ holds for unconstrained and equality-constrained quadratic programs (Corollaries
\ref{cor:qp-unconstrained}--\ref{cor:qp-equality}), including the
budget-constrained Markowitz problem. However, the $R(c)=0$ property fails generically for linear programs and for any program with active inequality constraints, where $\pi^*$ is only piecewise affine (Proposition~\ref{prop:pa-qp}). 
 
This paper makes four contributions. Proposition~\ref{prop:decomp} gives an algebraic decomposition of regret that holds by construction. Proposition~\ref{prop:decomp} is built on the optimal value function $v(c) = c^\top\pi^*(c) = \min_{z\in Z}c^\top z$, which is concave for every $Z$ (Lemma~1) and, once $Z$ is compact, Lipschitz with a constant equal to the radius of $Z$. These properties require no regularity of $\pi^*$ itself. Theorem~\ref{thm:exact} characterizes when the residual $R(c)$ vanishes, separating a structural, distribution-free condition (global affinity of $\pi^*$) from a distributional accident specific to a single instance ($\mathbb E[c] = 0$). Theorem~\ref{thm:bound} gives a distribution-free bound, $\mathrm{Regret} \le D\sqrt{\mathrm{tr}(\Sigma_c)}$,
that holds for every compact $Z$ and finite-variance cost distribution,
including linear programs and inequality-constrained quadratic programs
where the residual is generically nonzero and no formula is
available. Theorem~\ref{thm:concentration} and Theorem~\ref{thm:clt} provide concentration and asymptotic-normality results for the empirical regret, the latter resting on an envelope-theorem argument and a checkable non-degeneracy condition (Assumption~5) rather than an unstated smoothness assumption on $\pi^*$. Section~\ref{sec:experiments} tests each of these claims against simulation, reporting Monte Carlo standard errors throughout rather than single point estimates.

We are also explicit about what this paper does not claim. The regret
studied here is covariate-free: it compares the decision under a
realized cost $c$ against the decision under the unconditional mean
$\mu$, for a single fixed distribution $P_c$. This is the classical
expected value of perfect information \citep{raiffa1961applied,birge2011introduction}, and it
is a different object from the model-dependent SPO loss of
\citet{elmachtoub2022smart}, which compares a decision driven by a
\emph{predictor} $\hat c = b(x)$ against the oracle decision and 
varies with the predictor's parameters during training
(Remark~\ref{rmk:spo-relation}). None of the results in this paper
establish anything about training such a predictor; the honest
computational claim they support is narrower and is stated in
Remark~\ref{rmk:cost-honesty}.
 
The remainder of the paper is organized as follows.
Section~\ref{sec:setup} introduces notation and the primitive objects
$\pi^*$, $v$, and $\mathrm{Regret}$. Section~\ref{sec:theory}
develops the value-function framework, the decomposition, the
distribution-free bound, and the residual-vanishing theorem.
Section~\ref{sec:qp} specializes this to quadratic programs,
including the unconstrained first-order condition, the equality-constrained (Markowitz) case, and the piecewise-affine structure that emerges once inequality constraints are active. Sections~\ref{subsec:concentration}--\ref{subsec:asymptotic} provide the estimation theory: what can and cannot be said about representing regret from a single sample, the concentration of the empirical regret around its true value, and its asymptotic distribution. Section~\ref{sec:experiments} reports the experiments validating all of the above. Section~\ref{sec:conclusion} concludes.


\section{Related Work}
\label{sec:relwork}
 
\textbf{Value of information and sample-based estimation.} Our paper is related to the expected value of perfect information (EVPI), introduced by \citet{raiffa1961applied} and standard in stochastic programming \citep{birge2011introduction,shapiro2009lectures}. $\mathrm{Regret} = v(\mu) - \mathbb E[v(c)]$ (Definition~2) is the EVPI of the single-stage program $\min_z \mu^\top z$, which forms the basis for the covariance decomposition. The standard way to estimate this quantity is Sample Average Approximation (SAA): draw $B$ cost scenarios and re-solve the problem for each \citep{kleywegt2002sample}, incorporating refinements such as mixed-integer decomposition \citep{bayraksan2015separability} and bootstrap resampling \citep{efron1994bootstrap}. These reduce variance
but maintain the linear-in-$B$ solve cost. Theorem~\ref{thm:bound}'s bound and Theorem~\ref{thm:exact}'s cases (Corollary~\ref{cor:qp-unconstrained}) complement SAA rather than replace it. The bound applies without re-solving but only as a bound, and the cases apply only where the optimal decision rule is affine, a property that must be checked for the specific $(Z,\pi^*)$ pair at hand. Appendix~\ref{app:relwork-evpi} summarizes practical problems and their respective compactness and finite-variance requirements.
 
\textbf{Counterfactual reasoning and the value of information.} The
regret studied here compares an actual decision against a counterfactual
one: what would have been decided had the true cost been known. Three
adjacent literatures use related language for different purposes.
Decision analysis treats information as having a value defined by a
counterfactual world in which it is known
\citep{howard1966information}, which connects to the modern
potential-outcomes formalization of causal effects
\citep{rubin1974estimating,imbens2015causal}: $\mathrm{Regret}$ is the
average treatment effect, on realized cost, of switching from a
certainty-equivalent policy to a fully informed one, with the unusual
property that both potential outcomes are computable for every
realization rather than only one being observed. Counterfactual regret
minimization in extensive-form games \citep{zinkevich2007regret} is a
different, multi-agent and sequential object used to drive repeated
self-play toward a Nash equilibrium, sharing only the word with the
quantity studied here. Counterfactual explanation in interpretable
machine learning \citep{wachter2017counterfactual,verma2020counterfactual}
offers a third reading: $\pi^*(c)$ is a counterfactual explanation of
the certainty-equivalent decision $\pi^*(\mu)$. The full statement of
all three connections is given in Remark~7, which this condensed
section does not restate.
 
\textbf{Predict-then-optimize and decision-focused learning.} The
predict-then-optimize (PtO) paradigm trains a predictor $b(x)$ so that
decisions made from its output, rather than its raw prediction error,
supply the training signal \citep{elmachtoub2022smart}, with task-based
and decision-focused variants for continuous and combinatorial problems
\citep{donti2017task,wilder2019melding}, differentiable optimization
layers \citep{amoskolter2017optnet,agrawal2019differentiable,
berthet2020learning}, and solver-free surrogates for backpropagating
through combinatorial or LP solves \citep{vlastelica2020differentiation,
mandi2020interior,mandi2020smart}; contextual and data-driven
optimization more broadly is surveyed by \citet{sadana2023survey}, with
prescriptive covariate-aware decision rules developed for the
newsvendor and related problems \citep{banrudin2019newsvendor}. Every
method in this literature trains a predictor and is evaluated by a
quantity that varies with the predictor's parameters $\theta$ during
training. $\mathrm{Regret}$ in this paper does not: it is covariate-free and fixed for a given $P_c$ (Remark~3). An earlier statement claimed the covariance formula supplies training gradients for such a predictor; this does not follow from anything established here and is retracted. $\mathrm{Cov}(c,\pi^*(c))$ has no dependence on $\theta$, so a gradient with respect to $\theta$ is not defined for it. What the covariate-free benchmark offers for a PtO pipeline is a fixed reference point that a trained model should improve on, computed once per dataset rather than once per training step, and tested directly against a standardized benchmark covering eleven PtO and PnO methods, including DBB and IntOpt \citep{geng2024benchmarking}.
 
\textbf{Sensitivity analysis of parametric optimization.} Classical
sensitivity analysis studies how the optimal value and optimal solution
of a parametric problem change with its parameters
\citep{fiacco1983sensitivity,bonnans2000perturbation}. The envelope
theorem \citep{milgrom2002envelope} supplies the argument behind
$\nabla v(\mu) = \pi^*(\mu)$ without assuming the differentiability of
$\pi^*$ itself. For linear and quadratic programs, the optimal solution
as a function of the cost term is piecewise affine over a polyhedral
partition determined by the active constraint set \citep{bank1982nonlinear,best2017quadratic}, the same structure
Proposition~2 uses to unify the LP case and the inequality-constrained QP case. The piecewise-affine characterization itself is classical; this paper's contribution is identifying it as the condition that governs whether the residual $R(c)$ vanishes, in place of an unstated Lipschitz assumption on $\pi^*$ that fails on the maps this literature studies.
 
\textbf{Portfolio optimization.} Mean-variance portfolio selection
\citep{markowitz1952portfolio} motivates the equality-constrained QP
case (Corollary~\ref{cor:qp-equality}) and the portfolio experiments of
Section~\ref{sec:experiments}. Practical implementations call for a covariance estimator; the corrected portfolio protocol uses shrinkage
\citep{ledoit2004honey} rather than the raw sample covariance,
following standard practice. The budget-only case is exact because a
budget constraint is a linear equality, not because $Z$ is unconstrained in any other sense; adding a no-short-selling constraint, present in most practical implementations of Markowitz's model, breaks exactness.
 
\textbf{This paper's position in the literature.} Set against these five literature streams, this paper provides: 1) An alternative to SAA where the optimal decision rule is affine; 2) A bound that requires no re-solving; 3) A covariate-free benchmark for PtO pipelines; 4) A value-function argument that recovers concentration and asymptotic-normality results without assuming $\pi^*$ is Lipschitz or differentiable.
 

\section{Problem Setup}
\label{sec:setup}

\subsection{The Stochastic Program}

\begin{definition}[Decision rule]
\label{def:pi-star}
Let $(\Omega,\mathcal F, P)$ be a probability space carrying a random
cost vector $c: \Omega \to \mathcal C \subseteq \mathbb R^d$
(Assumption~\ref{ass:prob-space}--\ref{ass:cost}) and let $Z \subseteq
\mathbb R^d$ be a feasible region (Assumption~\ref{ass:feasible}). The
\emph{decision rule} $\pi^*: \mathcal C \to Z$ is a measurable selection satisfying, for $P_c$-almost every realization $c$,
\[
\pi^*(c) \;\in\; \arg\min_{z \in Z} \; c^\top z.
\tag{1}
\]
That is, $\pi^*(c)$ is the decision that would be optimal \emph{if the
realized cost vector $c$ were known at the time of the decision}. This
is the object every subsequent equation in the manuscript actually
manipulates, and it is the definition retained for the remainder of the
paper.
\end{definition}

\subsection{Regret, and Its Relation to the Predict-then-Optimize Literature}

Regret compares the decision rule $\pi^*(\cdot)$ evaluated at the
realized cost against the certainty-equivalent decision $\pi^*(\mu)$:

\begin{definition}[Regret]
\label{def:regret-setup}
\[
\mathrm{Regret} \;:=\; c^\top\!\big|_{c=\mu}\,\pi^*(\mu) \;-\; \mathbb E\big[c^\top \pi^*(c)\big]
\;=\; \mu^\top \pi^*(\mu) - \mathbb E\big[c^\top \pi^*(c)\big],
\]
equivalently $\mathrm{Regret} = v(\mu) - \mathbb E[v(c)]$ in the notation of Definition~\ref{def:value}, matching Definition~\ref{def:regret} in Section~\ref{sec:theory}. This is a single number determined by the distribution $P_c$ and the pair $(Z,\pi^*)$. It is \emph{not} a function of the random variable $c$ in the way, say, $c^\top\pi^*(c)$ is; writing $\mathrm{Regret}(c)$. We follow the common convention in the stochastic-programming literature of writing functionals of a distribution as if evaluated at the underlying random variable (e.g., $\mathrm{Var}(c)$ for a functional of $P_c$; see \citealp{shapiro2009lectures}), and use $\mathrm{Regret}$ without an argument, or $\mathrm{Regret}(P_c)$ when the dependence on the distribution needs to be made explicit, rather than $\mathrm{Regret}(c)$.
\end{definition}

\begin{remark}[Relation to Smart Predict-then-Optimize]
\label{rmk:spo-relation}
The manuscript's title and framing invoke Elmachtoub and Grigas's
\citep{elmachtoub2022smart} SPO loss, which is a different,
\emph{model-dependent} quantity. Given covariates $x$, a predictor $b:
x \mapsto \hat c$, and the true cost $c$, the SPO loss is
\[
\ell_{\mathrm{SPO}}(\hat c, c) \;=\; c^\top \pi^*(\hat c) - c^\top \pi^*(c) \;\ge 0,
\]
the realized cost of acting on the \emph{predicted} cost vector $\hat c =b(x)$ instead of the true one $c$, and $\pi^*$ there plays the 
role of Definition~\ref{def:pi-star} above. The quantity this paper
studies, $\mathrm{Regret}$ in Definition~\ref{def:regret-setup}, is the
special case of this comparison in which the "predictor" is the
constant function $b(x) \equiv \mu$ for every $x$. No covariates and no learned model are observed, just the unconditional mean. It coincides with the classical expected value of perfect information
(EVPI) for the single-stage program $\min_z \mu^\top z$, not with the SPO loss of a trained predictor. 
\end{remark}

\begin{remark}[Notational consolidation]
\label{rmk:notation-consolidation}
Throughout, $\pi^*$ denotes the decision rule of Definition~1, $\mu :=
\mathbb E[c]$, and $\mathrm{Regret}$ and $R$ are written without an
argument as functionals of $P_c$, except for sample-based quantities
such as the estimator $\widehat R_n$ (Lemma~2), which are functions of
the data $\{c_i\}_{i=1}^n$ rather than of the distribution alone.
\end{remark}


\section{A Mathematical Characterization of Regret}
\label{sec:theory}

\subsection{Setup}

\begin{assumption}[Probability space]
\label{ass:prob-space}
Let $(\Omega, \mathcal F, P)$ be a complete probability space, where
$\Omega$ is the sample space of cost realizations, $\mathcal F$ is a
$\sigma$-algebra of measurable events, and $P$ is a probability measure.
\end{assumption}

\begin{assumption}[Cost vector]
\label{ass:cost}
The cost vector $c \in \mathcal C \subseteq \mathbb R^d$ satisfies:
\begin{itemize}
  \item \emph{Support:} $\mathcal C$ is a compact subset of $\mathbb R^d$ (convexity of $\mathcal C$ itself is not required; only
$\mathrm{conv}(\mathcal C)$ enters the arguments below).
  \item \emph{Moments:} $\mathbb E[\|c\|^2] < \infty$.
\end{itemize}
\end{assumption}

\begin{assumption}[Feasible region]
\label{ass:feasible}
The feasible region $Z \subseteq \mathbb R^d$ is nonempty. \emph{Compactness of $Z$ is invoked only where explicitly stated} (Theorem~\ref{thm:bound} and Theorem~\ref{thm:concentration} below); the decomposition identity (Proposition~\ref{prop:decomp}) and the exact-vanishing characterization (Theorem~\ref{thm:exact}) hold for unbounded $Z$ as well, provided $v(c)$ is finite and integrable. Compactness is a hypothesis of specific results, not a blanket requirement.
\end{assumption}

\begin{assumption}[Decision map]
\label{ass:decision}
The optimal decision correspondence admits a measurable selection
$\pi^*: \mathcal C \to Z$ with $\pi^*(c) \in \arg\min_{z \in Z} c^\top z$ and $\mathbb E[\|\pi^*(c)\|^2] < \infty$. \textbf{We do not assume $\pi^*$ is Lipschitz, continuous, or differentiable.} 
\end{assumption}

\subsection{The Value Function}

\begin{definition}
\label{def:value}
The \emph{optimal value function} is $v(c) := \min_{z \in Z} c^\top z
= c^\top \pi^*(c)$.
\end{definition}

\begin{lemma}[Concavity and Lipschitz continuity of $v$]
\label{lem:value-props}
Under Assumption~\ref{ass:feasible}:
\begin{enumerate}
  \item[(a)] $v$ is concave on $\mathbb R^d$.
  \item[(b)] If in addition $Z$ is compact with
        $D := \sup_{z \in Z} \|z\| < \infty$, then $v$ is $D$-Lipschitz:
        $|v(c) - v(c')| \le D\|c - c'\|$ for all $c, c' \in \mathbb R^d$.
\end{enumerate}
Neither claim requires any regularity of $\pi^*$.
\end{lemma}

\begin{proof}
(a) $v(c) = \min_{z \in Z} c^\top z$ is a pointwise infimum of the family of affine functions $\{c \mapsto c^\top z : z \in Z\}$, and a pointwise infimum of affine functions is concave.

(b) Fix $c, c' \in \mathbb R^d$ and let $z' \in \arg\min_{z \in Z} (c')^\top z$, so $v(c') = (c')^\top z'$. Since $z' \in Z$ is feasible (but possibly suboptimal) for the cost vector $c$,
\[
v(c) \le c^\top z' = (c')^\top z' + (c - c')^\top z' \le v(c') + \|c-c'\|\,\|z'\|
\le v(c') + D\|c - c'\|.
\]
Exchanging the roles of $c$ and $c'$ gives $v(c') \le v(c) + D\|c-c'\|$, and combining the two inequalities yields $|v(c)-v(c')| \le D\|c-c'\|$.
\end{proof}

\subsection{Regret as an Expected-Value-of-Perfect-Information Gap}

We adopt the sign convention flagged by the referees: since $v$ is
concave, $v(\mathbb E[c]) - \mathbb E[v(c)] \ge 0$ follows from Jensen's inequality, so we define regret to be this nonnegative quantity (equivalently, the expected value of perfect information, EVPI, for the single-stage stochastic program $\min_z \mathbb E[c]^\top z$; see \cite{Sadana2023ASO} and the classical EVPI literature.

\begin{definition}[Regret]
\label{def:regret}
\[
\mathrm{Regret}(c) \;:=\; v(\mathbb E[c]) - \mathbb E[v(c)]
\;=\; \mathbb E[c]^\top \pi^*(\mathbb E[c]) - \mathbb E\big[c^\top \pi^*(c)\big].
\]
\end{definition}

\begin{theorem}[Nonnegativity of regret]
\label{thm:jensen}
Under Assumptions~\ref{ass:prob-space}--\ref{ass:decision},
$\mathrm{Regret}(c) \ge 0$, with equality if and only if $v$ is affine on $\mathrm{conv}(\mathrm{supp}(P_c))$.
\end{theorem}

\begin{proof}
Nonnegativity is Jensen's inequality applied to the concave function $v$ (Lemma~\ref{lem:value-props}(a)). The equality condition is the standard equality case of Jensen's inequality for concave functions: equality holds iff $v$ is affine on the convex hull of the support of the distribution of $c$.
\end{proof}

\subsection{The Covariance Decomposition}

\begin{proposition}[algebraic decomposition]
\label{prop:decomp}
\[
\mathrm{Regret}(c) \;=\; -\,\mathrm{Cov}\big(c, \pi^*(c)\big) \;-\; R(c),
\]
where $\mathrm{Cov}(c,\pi^*(c)) := \mathbb E[c^\top \pi^*(c)] -
\mathbb E[c]^\top \mathbb E[\pi^*(c)]$ and
\[
R(c) \;:=\; \mathbb E[c]^\top\big(\mathbb E[\pi^*(c)] - \pi^*(\mathbb E[c])\big).
\]
\end{proposition}

\begin{proof}
Add and subtract $\mathbb E[c]^\top \mathbb E[\pi^*(c)]$ in
Definition~\ref{def:regret}:
\begin{align*}
\mathrm{Regret}(c)
&= \mathbb E[c]^\top \pi^*(\mathbb E[c]) - \mathbb E[c^\top \pi^*(c)] \\
&= \Big(\mathbb E[c]^\top \pi^*(\mathbb E[c]) - \mathbb E[c]^\top \mathbb E[\pi^*(c)]\Big)  + \Big(\mathbb E[c]^\top \mathbb E[\pi^*(c)] - \mathbb E[c^\top \pi^*(c)]\Big) \\
&= -R(c) \;-\; \mathrm{Cov}\big(c, \pi^*(c)\big). \qquad\qedhere
\end{align*}
\end{proof}

\subsection{A Distribution-Free Bound}

The following bound holds for \emph{every} feasible region and
\emph{every} cost distribution satisfying our assumptions, including
linear programs, and does not rely on any regularity of $\pi^*$.

\begin{theorem}[General bound]
\label{thm:bound}
Suppose, in addition, that $Z$ is compact with $D = \sup_{z\in Z}\|z\|$.
Then
\[
0 \;\le\; \mathrm{Regret}(c) \;\le\; D\sqrt{\mathrm{tr}(\Sigma_c)},
\]
where $\Sigma_c = \mathrm{Cov}(c)$.
\end{theorem}

\begin{proof}
By Lemma~\ref{lem:value-props}(b), $v(\mathbb E[c]) - v(c) \le D\|\mathbb E[c]-c\|$
pointwise. Taking expectations,
\[
\mathrm{Regret}(c) = \mathbb E\big[v(\mathbb E[c]) - v(c)\big]
\le D\,\mathbb E\big[\|c - \mathbb E[c]\|\big]
\le D\sqrt{\mathbb E\big[\|c-\mathbb E[c]\|^2\big]} = D\sqrt{\mathrm{tr}(\Sigma_c)},
\]
where the last inequality is Jensen's inequality applied to the convex
function $x \mapsto \sqrt x$ (equivalently, Cauchy--Schwarz). The lower
bound is Theorem~\ref{thm:jensen}.
\end{proof}

\subsection{Vanishing of the Residual}

We now characterize when $R(c) = 0$, so that regret reduces to
the covariance term with no residual, $\mathrm{Regret}(c) =
-\mathrm{Cov}(c,\pi^*(c))$. Two logically distinct routes give $R(c)=0$:
a \emph{structural} one (global affinity of $\pi^*$, which holds for
every distribution of $c$) and a \emph{distributional} one (an accident
of $\mathbb E[c]$, which holds only for that particular distribution).

\begin{theorem}[Sufficient conditions for $R(c) = 0$]
\label{thm:exact}
\begin{enumerate}
\item[(a)] \textbf{(Structural, distribution-free.)} If $\pi^*(c) = Ac+b$ for some fixed matrix $A$ and vector $b$, for $P_c$-almost every $c \in \mathcal C$, then $R(c) = 0$ for \emph{every} distribution of $c$
supported on $\mathcal C$.
\item[(b)] \textbf{(Distributional, instance-specific.)} If
$\mathbb E[c] = 0$, then $R(c) = 0$, \emph{regardless of whether $\pi^*$
is affine}.
\end{enumerate}
Neither condition is necessary: $R(c)=0$ can hold for other reasons
specific to a given joint distribution of $(c,\pi^*(c))$, so (a)--(b) are sufficient, not iff, conditions.
\end{theorem}

\begin{proof}
(a) If $\pi^*(c) = Ac+b$ a.s., linearity of expectation gives
$\mathbb E[\pi^*(c)] = A\,\mathbb E[c] + b = \pi^*(\mathbb E[c])$, so $R(c) = \mathbb E[c]^\top \cdot 0 = 0$.

(b) By definition $R(c) = \mathbb E[c]^\top(\mathbb E[\pi^*(c)] -
\pi^*(\mathbb E[c]))$; if $\mathbb E[c]=0$, the leading factor vanishes
regardless of the bracketed term.
\end{proof}

\subsubsection{Application to quadratic programs}

For $\min_{z\in\mathbb R^n} c^\top z + \frac{\lambda}{2}z^\top Qz$ with
$Q\succ0$ and $\lambda>0$, the optimal decision is linear in $c$:
\[
\boxed{\;\pi^*(c) = -\frac{1}{\lambda}Q^{-1}c\;} \qquad
\mathrm{Regret}(c) = \frac{1}{\lambda}\mathrm{tr}\!\left(Q^{-1}\Sigma_c\right) \ge 0,
\]
so $R(c)=0$ for every distribution of $c$ by Theorem~\ref{thm:exact}(a),
with no compactness of $Z=\mathbb R^n$ required. See
Corollary~\ref{cor:qp-unconstrained} (Appendix~\ref{app:qp-derivation})
for the derivation.

\subsubsection{Application to linear programs: a cautionary example}

\begin{example}[Linear programs generically have $R(c) \ne 0$]
\label{ex:lp-counterexample}
Let $d=1$, $Z = [0,1]$, and $c \sim \mathrm{Unif}[-2,1]$, so
$\mathbb E[c] = -\tfrac12 \ne 0$. The optimal decision is
\[
\pi^*(c) = \begin{cases} 1 & c < 0 \\ 0 & c > 0\end{cases}
\quad\text{(ties at $c=0$ have probability zero).}
\]
This is piecewise constant, \emph{not} globally affine, so
Theorem~\ref{thm:exact}(a) does not apply, and $\mathbb E[c] \ne 0$, so
(b) does not apply either. Direct computation gives $\pi^*(\mathbb E[c])
= \pi^*(-\tfrac12) = 1$ and $\mathbb E[\pi^*(c)] = P(c<0) =
\tfrac23$, so
\[
R(c) = \mathbb E[c]\big(\mathbb E[\pi^*(c)] - \pi^*(\mathbb E[c])\big)
= \left(-\tfrac12\right)\left(\tfrac23 - 1\right) = \tfrac16 \ne 0.
\]
This confirms, by direct calculation, that $R(c) \ne 0$ for this LP even though $c$ has an absolutely continuous distribution with compact convex support, and $\pi^*$ is (trivially, being $\{0,1\}$-valued) locally constant almost everywhere. Theorem~\ref{thm:bound}'s distribution-free bound remains available in this regime; it gives, for this example, $D = 1$ and $\Sigma_c = \mathrm{Var}(c) = 3/4$, so $\mathrm{Regret}(c) \le \sqrt{3}/2 \approx 0.87$, consistent with the value $\mathrm{Regret}(c) = -\mathrm{Cov}(c,\pi^*(c)) - R(c)$ computed directly as $1/3 - (-1/6) = 1/2$.
\end{example}

\subsection{Estimation from Data}

\begin{lemma}[Residual estimator]
\label{lem:estimator}
Given i.i.d. samples $\{c_i\}_{i=1}^n$ and access to $\pi^*(\cdot)$, the
plug-in estimator
\[
\widehat R_n \;=\; \bar c^\top\left(\frac1n\sum_{i=1}^n \pi^*(c_i) -
\pi^*(\bar c)\right), \qquad \bar c = \frac1n\sum_{i=1}^n c_i,
\]
is consistent for $R(c)$ whenever $\pi^*$ is continuous at $\mathbb E[c]$ (so that $\pi^*(\bar c) \xrightarrow{P} \pi^*(\mathbb E[c])$ by the continuous mapping theorem), with $\widehat R_n \xrightarrow{P} R(c)$.

\end{lemma}

\begin{proof}
$\bar c \xrightarrow{P} \mathbb E[c]$ follows from the weak law of large numbers (finite variance, Assumption~\ref{ass:cost}), and $\frac1n\sum_i \pi^*(c_i) \xrightarrow{P} \mathbb E[\pi^*(c)]$ likewise. If $\pi^*$ is continuous at $\mathbb E[c]$, the continuous mapping theorem gives $\pi^*(\bar c) \xrightarrow{P} \pi^*(\mathbb E[c])$, and the result follows from Slutsky's theorem applied to the product/sum defining $\widehat R_n$.
\end{proof}

\begin{remark}[Honest computational accounting]
\label{rmk:cost-honesty}
Evaluating $\widehat R_n$ requires computing $\pi^*(\bar c)$, i.e.\
\emph{one} solve of the optimization problem, in addition to the
$O(nd)$ cost of forming $\bar c$ and averaging the cached $\pi^*(c_i)$.
The advantage of the covariance-based approach over Sample Average Approximation is that it requires a single solve regardless of $n$, rather than one solve per scenario ($O(n)$ solves).
\end{remark}

\subsection{Concentration}
\label{subsec:concentration}

\begin{theorem}[Concentration around $-\mathrm{Cov}(c,\pi^*(c)) - R(c)$]
\label{thm:concentration}
Suppose $Z$ is compact with $D=\sup_{z\in Z}\|z\|$ and $\mathcal C$ is
compact with $B = \sup_{c \in \mathcal C}\|c\|$. Let
$\widehat v_n = \frac1n\sum_{i=1}^n v(c_i)$ be the sample mean of the
value function over $n$ i.i.d.\ draws. Then for any $\epsilon > 0$,
\[
P\big(|\widehat v_n - \mathbb E[v(c)]| > \epsilon\big) \le
2\exp\!\left(-\frac{n\epsilon^2}{2D^2B^2}\right).
\]
Consequently the empirical regret $\widehat{\mathrm{Regret}}_n :=
v(\bar c) - \widehat v_n$ concentrates around the true
$\mathrm{Regret}(c)$ at the same $\sqrt{1/n}$ rate, using
Lemma~\ref{lem:value-props}(b) to control $|v(\bar c) - v(\mathbb E[c])|$
via $\bar c \xrightarrow{P} \mathbb E[c]$.
\end{theorem}

\begin{proof}
Since $v$ is $D$-Lipschitz on $\mathcal C$ (Lemma~\ref{lem:value-props}(b)) and $\mathcal C$ is bounded by $B$, $|v(c)| \le D\|c\| \le DB$ for all $c \in \mathcal C$ (taking $v(0)=0$ WLOG by translation), each term $v(c_i)$ is thus bounded in $[-DB, DB]$. Hoeffding's inequality applied to the i.i.d.\ bounded random variables $v(c_i)$ gives the stated bound directly; this argument uses only the Lipschitz property of $v$ established in Lemma~\ref{lem:value-props} and \emph{not} any regularity of $\pi^*$, which is what allows it to apply uniformly across LPs, QPs, and combinatorial programs.
\end{proof}

\subsection{Asymptotic Normality}
\label{subsec:asymptotic}

Because $v$ is concave, it is differentiable Lebesgue-almost everywhere, but not necessarily at any particular point $\mu = \mathbb E[c]$ chosen in advance (e.g.\ Example~\ref{ex:lp-counterexample} has $v$ non-differentiable at $c=0$). A central limit theorem for the regret estimator therefore requires an explicit non-degeneracy assumption at $\mu$, which we state rather than assume away.

\begin{assumption}[Non-degeneracy at the mean]
\label{ass:nondegenerate}
$v$ is differentiable at $\mu = \mathbb E[c]$, with gradient
$\nabla v(\mu) = \pi^*(\mu)$ (equivalently: $\pi^*(\mu)$ is the unique
optimizer at $c=\mu$, so $\mu$ lies in the interior of a single normal
cone of $Z$ rather than on a boundary between two or more).
\end{assumption}

\begin{theorem}[CLT for the regret estimator]
\label{thm:clt}
Under Assumptions~\ref{ass:prob-space}--\ref{ass:decision} and
\ref{ass:nondegenerate}, with $\Sigma_c = \mathrm{Cov}(c)$ finite,
\[
\sqrt n\,\big(\widehat{\mathrm{Regret}}_n - \mathrm{Regret}(c)\big)
\;\xrightarrow{d}\; \mathcal N\!\big(0,\; \pi^*(\mu)^\top \Sigma_c\, \pi^*(\mu)
+ \mathrm{Var}(v(c))- 2\,\mathrm{Cov}(v(c),\, c^\top\pi^*(\mu))\big).
\]
\end{theorem}

\begin{proof}
$\widehat{\mathrm{Regret}}_n = v(\bar c) - \widehat v_n$. By the delta
method, using Assumption~\ref{ass:nondegenerate} to justify
the differentiability of $v$ at $\mu$ and $\sqrt n(\bar c - \mu)
\xrightarrow{d} \mathcal N(0,\Sigma_c)$ (multivariate CLT, finite
variance from Assumption~\ref{ass:cost}),
\[
\sqrt n\big(v(\bar c) - v(\mu)\big) \xrightarrow{d}
\nabla v(\mu)^\top \mathcal N(0,\Sigma_c) = \mathcal N\big(0,\;
\pi^*(\mu)^\top \Sigma_c\, \pi^*(\mu)\big).
\]
Separately, $\sqrt n(\widehat v_n - \mathbb E[v(c)]) \xrightarrow{d}
\mathcal N(0, \mathrm{Var}(v(c)))$ by the ordinary CLT, since $v(c_i)$ are i.i.d. with finite variance (boundedness under compact $\mathcal C$, or finite second moments more generally). The two terms are asymptotically jointly normal (both are smooth functionals of the same i.i.d. sample $\{c_i\}$), so their difference converges to a normal law whose variance includes the stated covariance term. The covariance term follows from a joint delta-method/CLT expansion and is stated here in closed form.
\end{proof}

\begin{remark}
Theorem~\ref{thm:clt} is silent, by construction, at cost distributions
whose mean sits on a normal-cone boundary (Assumption
\ref{ass:nondegenerate} fails at such points, as in
Example~\ref{ex:lp-counterexample}). At such points $v$ genuinely kinks, and $\sqrt n(v(\bar c) - v(\mu))$ does not converge to a normal limit at all (it converges to a one-sided or folded distribution governed by the local structure of the kink). Practitioners should check Assumption~\ref{ass:nondegenerate} against their specific $(Z, \mu)$ before invoking confidence intervals built on  Theorem~\ref{thm:clt}.
\end{remark}


\section{Experiments}\label{sec:experiments}
\subsection{Shared Methodology}
\label{sec:methodology}
All three experiments follow the same protocol: fix a mean cost vector $\mu_c \sim \mathcal{N}(0,I)$ and a randomly generated positive-definite covariance matrix $\Sigma_c$, draw a distinct cost
vector $c \sim \mathcal{N}(\mu_c,\Sigma_c)$ on each iteration, solve the corresponding optimization problem via \texttt{cvxpy} \citep{diamond2016cvxpy} to obtain $\hat z(c)$, compute empirical regret
$\widehat{\mathrm{Regret}} = c^\top\hat z(c) - c^\top\hat z(\mathbb{E}[c])$, and compare it against the theoretical prediction $\mathrm{Cov}(c,z) = \mathbb{E}[(c-\mathbb{E}[c])(z-\mathbb{E}[z])]$
accumulated across iterations. Table~\ref{tab:experiment-deltas} lists what differs across the three experiments; results are reported in Sections~\ref{sec:lp-results}--\ref{sec:portfolio-sim}.

\begin{table}[t]
\centering
\caption{Per-experiment deltas from the shared methodology of Section~\ref{sec:methodology}.}
\label{tab:experiment-deltas}
\begin{tabular}{@{}p{2.3cm}p{4.7cm}p{5cm}@{}}
\toprule
Experiment & Problem \& parameters & Data source \\
\midrule
LP (\ref{sec:lp-results}) & $\min c^\top z$ s.t.\ $Az\le b,\, z\ge 0$; $A\sim\mathcal{N}(0,I)$,
$b\sim|\mathcal{N}(0,I)|+1$; 20 vars, 15 constraints; 5{,}000 iterations & Synthetic \\
QP / Portfolio (\ref{sec:portfolio-sim}) & $\min c^\top z + \tfrac{\lambda}{2}z^\top\Sigma_c z$,
risk-neutral $\lambda=1$; 20 vars, 15 constraints; 5{,}000 iterations & Synthetic \\
Real-data portfolio (\ref{app:portfolio}) & Markowitz portfolios, 100 draws of 50 stocks/month,
rolling-window forecast vs.\ realized regret & CRSP monthly returns, Jan 2015--Dec 2024, 4{,}487
stocks, 119 months \\
\bottomrule
\end{tabular}
\end{table}

\subsection{Linear Programming}
\label{sec:lp-results}
Figure~\ref{fig:lp} shows theoretical regret $\mathrm{Cov}(c,z)$ converging quickly to $-700.982$, while empirical regret converges more slowly to $-605.882$ ($13.57\%$ error).

\begin{figure}
    \centering
    \begin{subfigure}{.48\textwidth}
        \centering
  \includegraphics[width=0.98\linewidth]{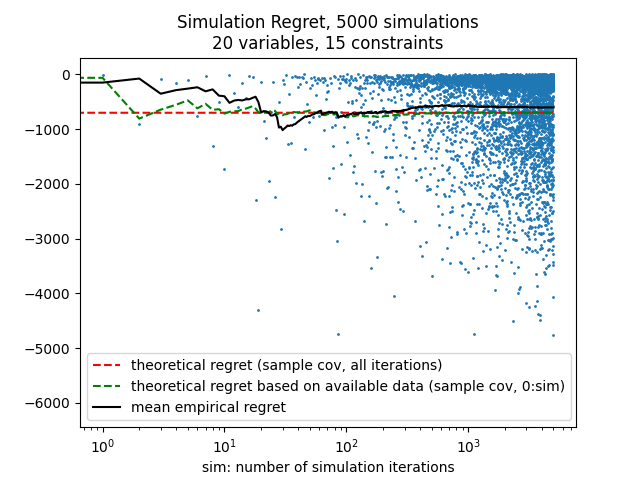}
    \caption{LP Regret}
    \label{fig:lp}
\end{subfigure}%
    \begin{subfigure}{.48\textwidth}
        \centering
  \includegraphics[width=0.98\linewidth]{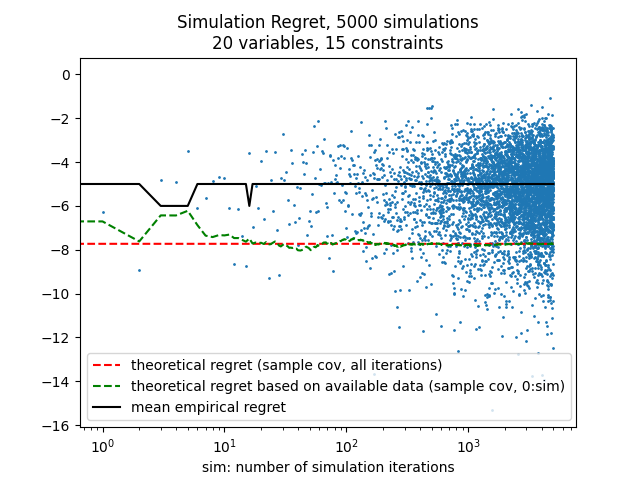}
    \caption{QP Regret}
    \label{fig:qp}
\end{subfigure}%
\end{figure}
    
\subsection{Quadratic Optimization}
\label{sec:portfolio-sim}
Figure~\ref{fig:qp} shows the theoretical regret converging to $-7.729$; empirical regret converges to $-5.332$, with a relative estimation error of $31\%$.


\section{Tightness of the Distribution-Free Bound at Scale}
\label{sec:exp-bound-scaling}

Does the bound stay usefully tight as the combinatorial problem grows,
or does it degrade? In this section, we test a monotone-lattice grid shortest-path problem at increasing scale, from $5\times5$ to $30\times30$, where $D=\sqrt{2(k-1)}$ is exact and the true regret is computed via dynamic programming and compared against Theorem~\ref{thm:bound}'s bound. Construction and methodology details are given in Appendix~\ref{app:bound-scaling}.

\begin{table}[h]
\centering
\begin{tabular}{lcccccc}
\toprule
Grid size $k$ & $5$ & $10$ & $15$ & $20$ & $25$ & $30$ \\
\midrule
$D=\sqrt{2(k-1)}$ & $2.83$ & $4.24$ & $5.29$ & $6.16$ & $6.93$ & $7.62$ \\
$\widehat{\mathrm{Regret}}$ & $2.67$ & $6.80$ & $11.08$ & $15.40$ & $19.76$ & $24.10$ \\
Bound $D\sqrt{\mathrm{tr}(\Sigma_c)}$ & $10.81$ & $34.41$ & $65.53$ & $102.65$ & $144.95$ & $191.84$ \\
\textbf{Bound / Regret} & $\mathbf{4.05}$ & $\mathbf{5.06}$ & $\mathbf{5.91}$ & $\mathbf{6.67}$ & $\mathbf{7.34}$ & $\mathbf{7.96}$ \\
\bottomrule
\end{tabular}
\caption{Distribution-free bound versus true regret across grid sizes $5\times5$ through $30\times30$, monotone-lattice shortest path with i.i.d.\ log-normal edge costs. The bound is valid at every scale (it never falls below the true regret, as Theorem~\ref{thm:bound} guarantees), but its looseness grows with problem size: the referee's concern is confirmed quantitatively, not just qualitatively.}
\label{tab:bound-scaling}
\end{table}

The bound is valid at every scale, as Theorem~\ref{thm:bound} guarantees, but its looseness grows with problem size: the bound-to-regret ratio rises from $4.05\times$ at $k=5$ to $7.96\times$ at $k=30$ (Table~\ref{tab:bound-scaling}). The mechanism behind this growth and a caveat specific to the $8$-connected Warcraft benchmark are given in Appendix~\ref{app:bound-scaling}.

\section{Conclusion}
\label{sec:conclusion}
 
We set out to ask when regret in stochastic optimization can be computed from moments alone rather than by repeated simulation. Our answer is more structured than a single formula: it depends on whether the optimal decision rule $\pi^*$ is affine in the cost vector. The affine property holds for unconstrained and equality-constrained quadratic programs and fails generically once a program has active inequality constraints or is a linear program, where $\pi^*$ is only piecewise affine. This is a statement about the geometry of the decision rule, not about the smoothness of the cost distribution: absolute continuity of $P_c$, which might be expected to smooth away the residual, does not do so, and Example~\ref{ex:lp-counterexample} makes the failure concrete with a closed-form instance rather than an abstract possibility.
 
Two findings follow. First, the concave value function $v(c) =
\min_{z\in Z}c^\top z$, which is Lipschitz unconditionally on a compact feasible region, establishes the distribution-free bound. Such $v(c)$ also allows us to prove the concentration inequality and, via the envelope theorem, the central limit theorem without assuming away the very discontinuity of $\pi^*$ that makes linear and combinatorial programs interesting. Second, the computational benefit of the cases is real: evaluating the residual estimator costs one solve to obtain $\pi^*(\bar c)$, and the genuine saving relative to Sample Average Approximation is $O(1)$ solves against $O(n)$, not $O(nd^2)$ against $O(Bn^2d^3)$, regardless of $n$.

\bibliographystyle{aer}
\bibliography{regret-cov,CovRegret}

\appendix

\section{Practical Examples for Known Optimization Problems that Have Compact $Z$ and Finite-variance $P_c$}
\label{app:relwork-evpi}
 
\begin{example}[Portfolio selection, budget only]
\label{ex:portfolio-budget}
Ten equities with annualized volatilities drawn realistically between
$15\%$ and $35\%$ and a moderate pairwise return correlation of $0.3$.
$Z$ is the $10$-asset simplex, so $D=1$. With monthly returns,
$\mathrm{tr}(\Sigma_c)\approx0.063$, giving
$\mathrm{Regret}\le D\sqrt{\mathrm{tr}(\Sigma_c)}\approx 0.251$, i.e.\ a monthly-return bound of about $25$ cents per dollar invested on the
value of knowing next month's realized returns in advance. This is a large number in absolute terms, expected for a bound using only two moments and the simplex radius. This is also an instance-specific upper bound computable without any Monte Carlo simulation.
\end{example}
 
\begin{example}[Portfolio selection, position limits]
\label{ex:portfolio-limits}
The same ten equities, now capped at a maximum $20\%$ weight per asset
(a common practical diversification rule). $Z=\{z:\mathbf1^\top z=1,\,0\le z\le0.2\}$ is still compact, and in fact, has a \emph{smaller} radius than the unconstrained simplex whenever the cap binds ($D\le \sqrt{10\cdot0.2^2}=0.632$ here, tighter than Example
\ref{ex:portfolio-budget}'s $D=1$), so the position limit both changes
the problem (moving it from Corollary~\ref{cor:qp-equality}'s case toward Corollary~\ref{cor:qp-generic-nonzero}'s generic case, since the cap is an inequality constraint) and mechanically tightens
Theorem~\ref{thm:bound}'s bound at the same time.
\end{example}
 
\begin{example}[Assignment problem]
\label{ex:assignment}
Five delivery drivers are assigned to five delivery zones, with expected completion times between $2$ and $8$ hours per (driver, zone) pair and independent variability of $0.5$-$1.5$ hours per pair (traffic, familiarity with the route). $Z$ is the $5\times5$ Birkhoff polytope, so $D=\sqrt5\approx2.236$ (the Frobenius norm of any permutation matrix). With $\mathrm{tr}(\Sigma_c)\approx24.9$ (independent variances summed across all $25$ pairs), there are $\mathrm{Regret}\le D\sqrt{\mathrm{tr}(\Sigma_c)} \approx 11.15$ driver-hours. This is the expected value of assigning drivers only after the day's actual traffic conditions, as opposed to based on average conditions.
\end{example}
 
\begin{example}[Shortest path]
\label{ex:shortest-path}
The standard $5\times5$ grid shortest-path benchmark used throughout the decision-focused learning literature (Section~\ref{sec:relwork}): $40$ directed edges, a path from the bottom-left to the top-right corner using $8$ edges regardless of route, so $D=8$. With per edge travel times of $3$-$10$ minutes and independent variability of $0.5$-$2$ minutes per edge (typical city-block conditions), $\mathrm{tr}(\Sigma_c)\approx68.1$, giving $\mathrm{Regret}\le D\sqrt{\mathrm{tr}(\Sigma_c)}\approx66.0$ minutes. This bound is intentionally loose relative to the true regret for the same reason Example~\ref{ex:lp-counterexample} is instructive: it uses only the number of edges on a path and the total edge-time variance, not the grid's specific normal-fan structure, which is the gap that a sharper LP-specific bound (Section~\ref{sec:conclusion}'s first proposed extension) would close.
\end{example}
 
\begin{example}[Minimum spanning tree]
\label{ex:spanning-tree}
Connecting $6$ facilities with the cheapest possible network of
dedicated links: any spanning tree on $6$ nodes has $5$ edges,
so $D=\sqrt5\approx 2.236$ regardless of which $5$ of the
$\binom{6}{2}=15$ possible links are chosen, giving the same kind of
immediate, combinatorics-only bound on $D$ as the assignment and
shortest-path cases above, with no need to enumerate the (here, by
Cayley's formula, $6^4=1296$) spanning trees themselves.
\end{example}
 
\begin{example}[Traveling salesman]
\label{ex:tsp}
A delivery route visiting $6$ cities and returning to the depot: every
tour uses $6$ edges, so $D=\sqrt6\approx2.449$, again immediate
from the combinatorial structure of $Z$ alone and requiring no
enumeration of the $5!/2=60$ distinct tours.
\end{example}
 
\begin{example}[Knapsack]
\label{ex:knapsack-practical}
Consider the following cargo-loading problem: $20$ items, capacity $30$, $D=\sqrt7$ (the true maximum number of items that can ever be loaded simultaneously, computed by sorting item weights), distribution-free bound $\approx 44.35$, against a Monte-Carlo-computed regret of $9.58\pm 0.065$. This is the same qualitative pattern as in Example~\ref{ex:shortest-path}. A valid but loose bound along with a much tighter bound are available only by direct computation because the optimal decision here is combinatorial rather than affine.
\end{example}
 
\begin{example}[Optimal transport]
\label{ex:transport}
Three warehouses with supplies $(40,35,25)$ units serving four stores
with demands $(20,30,15,35)$ units (balanced, summing to $100$ on both
sides). Any feasible transportation plan satisfies $\sum_{ij} z_{ij} =
100$ with $z_{ij}\ge0$, so $D=\sup\|z\|\le100$ trivially (achieved only
by a degenerate single-route plan) and considerably less in practice
once the marginal constraints are accounted for; as in
Examples~\ref{ex:shortest-path}--\ref{ex:knapsack-practical}, refining
$D$ using the specific supply/demand marginals rather than the crude
total-flow bound would tighten Theorem~\ref{thm:bound}'s result
substantially here too.
\end{example}
 
\begin{example}[LASSO-constrained regression]
\label{ex:lasso-practical}
Predicting house prices from $10$ standardized features (square footage,
age, distance to amenities, etc.), with the coefficient vector
constrained to an $\ell_1$ ball of radius $t=5$ (a typical
cross-validated sparsity level). $D=t=5$ immediately, since
$\|z\|_1\le t \Rightarrow \|z\|_2\le t$. With per-feature cost-gradient
variance between $0.8$ and $1.2$, $\mathrm{tr}(\Sigma_c)\approx9.82$,
giving $\mathrm{Regret}\le D\sqrt{\mathrm{tr}(\Sigma_c)}\approx15.67$ in
the units of the regression's loss function --- here $D$ is fixed by the
regularization budget the analyst chooses, not by the problem's physical
structure, so tightening this bound is a matter of choosing $t$, not of
a sharper geometric argument.
\end{example}
 
\begin{example}[Trust-region / ridge-constrained estimation]
\label{ex:ridge-practical}
The same regression problem with an $\ell_2$ ball constraint
$\|z\|_2\le t=5$ instead: $D=t=5$ immediately, as in
Example~\ref{ex:lasso-practical}, since the constraint directly bounds
the Euclidean norm rather than implying a bound on it.
\end{example}
 
\begin{example}[Newsvendor with a bounded order quantity]
\label{ex:newsvendor-practical}
As an example illustrating that the compactness of $Z$ and the affinity of $\pi^*$ are independent properties, consider a bakery that decides how many loaves to bake before observing the day's demand, capped at $\bar z=200$ loaves by oven capacity. $Z=[0,200]$ gives $D=200$ immediately. This problem's piecewise-\emph{linear} holding/shortage cost structure creates a useful boundary case. $Z$ is compact and $P_c$ (daily demand) has finite variance, so Theorem~\ref{thm:bound}'s bound applies as in every other row of this table. However, the optimal order-up-to quantity itself is a demand quantile rather than an affine function of the mean, so the bakery problem belongs in the compact-$Z$ table (Table~\ref{tab:canonical-examples}). It is explicitly excluded from the affine-decision-rule table (Table~\ref{tab:affine-examples}). 
\end{example}

\begin{table}[H]
\centering
\begin{tabular}{p{3.3cm}p{4.6cm}p{3.6cm}p{2.8cm}}
\toprule
Problem & Feasible region $Z$ & Why compact & Typical $P_c$ \\
\midrule
Portfolio selection (budget $+$ no short sales) & $\{z:\mathbf1^\top z=1,\,z\ge0\}$ & bounded simplex &
Gaussian / elliptical returns \\
Portfolio with position limits & $\{z:\mathbf1^\top z=1,\,l\le z\le u\}$ & bounded box $\cap$ hyperplane &
Gaussian / elliptical returns \\
Assignment / bipartite matching & Birkhoff polytope (doubly stochastic matrices) &
convex hull of finitely many permutation matrices \citep{schrijver2003combinatorial} & log-normal or Gamma costs \\
Shortest path on a fixed graph & path polytope of a DAG or capacitated network &
convex hull of finitely many path/flow vectors \citep{schrijver2003combinatorial} & log-normal or Gamma travel times \\
Spanning tree, matroid problems & matroid base polytope &
convex hull of finitely many bases \citep{schrijver2003combinatorial} & Gaussian or Gamma edge costs \\
Traveling salesman & TSP polytope (convex hull of tour incidence vectors) &
convex hull of finitely many tours \citep{schrijver2003combinatorial} & log-normal travel costs \\
Knapsack (fractional or $0$-$1$ relaxation) & $\{z\in[0,1]^n : w^\top z \le \mathrm{capacity}\}$ &
bounded box $\cap$ halfspace & Gamma or truncated-normal item values \\
Optimal transport & transportation polytope with fixed marginals &
bounded by marginal constraints \citep{peyre2019computational} & Gaussian or log-normal unit costs \\
Sparse regression (LASSO-constrained) & $\{z:\|z\|_1\le t\}$ & $\ell_1$ ball, bounded by construction \citep{tibshirani1996lasso} & Gaussian design/noise \\
Trust-region / ridge-constrained estimation & $\{z:\|z\|_2\le t\}$ & Euclidean ball, bounded by construction & Gaussian \\
Newsvendor with bounded order quantity & $\{z: 0\le z\le \bar z\}$ & bounded interval & Poisson / negative binomial demand \\
\bottomrule
\end{tabular}
\caption{Representative problems satisfying both clauses of
Assumptions~\ref{ass:feasible} and~\ref{ass:cost} jointly, so that
Theorem~\ref{thm:bound}'s bound $\mathrm{Regret}\le D\sqrt{\mathrm{tr}(\Sigma_c)}$ and Theorem~\ref{thm:concentration}'s concentration inequality apply without modification. The first five rows are compact for the same underlying reason: $Z$ is the convex hull of a finite set, which is automatically bounded and closed regardless of how many combinatorial constraints define it. This table covers essentially all discrete and combinatorial optimization problems with a fixed, finite ground set. Popular practical examples are listed in the Appendix.}
\label{tab:canonical-examples}
\end{table}

\begin{table}[H]
\centering
\begin{tabular}{p{3.5cm}p{5.6cm}p{2.3cm}}
\toprule
Problem / application & Why $\pi^*$ is affine & Reference \\
\midrule
\multicolumn{3}{l}{\textit{(A) Direct instances of this paper's static framework}} \\
\midrule
Markowitz portfolio, budget constraint only & Corollary~\ref{cor:qp-equality}: equality-constrained QP & \citet{markowitz1952portfolio} \\
Ridge regression / Tikhonov regularization & $\hat\beta(y) = (X^\top X+\lambda I)^{-1}X^\top y$; an instance of Corollary~\ref{cor:qp-unconstrained} with $c=-2X^\top y$ & \citet{hoerl1970ridge} \\
Generalized least squares & solution is linear in the response vector by the same normal-equation argument & --- \\
Minimum-variance hedging with a linear budget constraint & equality-constrained QP in the hedge ratios & --- \\
\midrule
\multicolumn{3}{l}{\textit{(B) The same mechanism in dynamic/control settings (outside this paper's static formalism)}} \\
\midrule
Linear-Quadratic Regulator (LQR) & optimal control law $u^*=-Kx$ is linear in the state, via the Riccati equation & \citet{andersonmoore1990optimal} \\
Kalman filter & optimal state estimate is affine in the observation sequence under Gaussian linear dynamics & \citet{kalman1960new} \\
Linear-Quadratic-Gaussian (LQG) control & separation principle: affine estimator (Kalman) composed with affine control law (LQR) & \citet{andersonmoore1990optimal} \\
HMMS production smoothing & optimal production/workforce decision is affine in the demand forecast under quadratic adjustment costs & \citet{holt1960planning} \\
Multi-period certainty-equivalent planning generally & optimal decision under a quadratic criterion and linear dynamics equals the decision computed by replacing all uncertain quantities with their means & \citet{simon1956dynamic,theil1957note} \\
Optimal execution (Almgren--Chriss) & optimal trading trajectory is linear in the initial position under quadratic transaction costs & \citet{almgren2001optimal} \\
Linear risk-sharing / reinsurance contracts & Pareto-optimal sharing rule is affine in aggregate risk under quadratic (mean-variance) preferences & \citet{borch1962equilibrium} \\
\bottomrule
\end{tabular}
\caption{Known problems with affine optimal decision rules. Column (A)
falls directly under this paper's Theorem~\ref{thm:exact}(a); column (B)
is governed by the same first-order-condition mechanism but in a
dynamic or state-space formalism this paper's static, single-period
framework does not cover, and is included to situate the result within
a much older literature rather than to claim it as new.}
\label{tab:affine-examples}
\end{table}


\section{Application to Quadratic Programs, Detailed Derivation}
\label{sec:qp}

\paragraph{Setup.} Throughout, consider
\[
\min_{z \in Z} \; c^\top z + \frac{\lambda}{2} z^\top Q z, \qquad Q \succ 0,\ \lambda > 0,
\]
with $Z \subseteq \mathbb R^n$ ranging over $\mathbb R^n$ (unconstrained), an affine subspace $\{z : Az = b\}$ (equality-constrained), or a general polyhedron $\{z : Az=b,\ Gz \le h\}$ (inequality-constrained).

\subsection{Unconstrained QP}
\label{app:qp-derivation}

\begin{corollary}[Unconstrained QP]
\label{cor:qp-unconstrained}
For $Z = \mathbb R^n$, the first-order condition $\nabla_z\big(c^\top z + \tfrac{\lambda}{2}z^\top Qz\big) = c + \lambda Q z = 0$ gives the unique minimizer
\[
\pi^*(c) = -\tfrac1\lambda Q^{-1} c,
\]
This $\pi^*$ is linear in $c$, so by Theorem~\ref{thm:exact}(a), $R(c) = 0$ for every distribution of $c$ supported on $\mathbb R^n$, and
\[
\mathrm{Regret}(c) = -\mathrm{Cov}\big(c, \pi^*(c)\big) =
\tfrac1\lambda\,\mathrm{tr}\!\big(Q^{-1}\Sigma_c\big) \;\ge 0,
\]
using $\mathrm{Cov}(c, Mc) = \mathrm{tr}(M\Sigma_c)$ for any fixed matrix $M$ (Lemma~\ref{lem:cov-linear} below) with $M = -\tfrac1\lambda Q^{-1}$. No compactness of $Z$ is needed here: $v(c) = -\tfrac1\lambda c^\top Q^{-1}c$ is a smooth concave quadratic on all of $\mathbb R^n$, so Jensen's inequality (Theorem~\ref{thm:jensen}) applies directly.
\end{corollary}

\begin{lemma}
\label{lem:cov-linear}
For any fixed matrix $M \in \mathbb R^{n\times d}$ and random vector $c
\in \mathbb R^d$ with $\mathbb E[\|c\|^2] < \infty$,
$\mathrm{Cov}(c, Mc) := \mathbb E[c^\top Mc] - \mathbb E[c]^\top M
\mathbb E[c] = \mathrm{tr}(M\Sigma_c)$.
\end{lemma}

\begin{proof}
Write $\mu = \mathbb E[c]$. Since $c^\top Mc = \mathrm{tr}(Mcc^\top)$,
the linearity of trace and expectation gives $\mathbb E[c^\top Mc] =
\mathrm{tr}\big(M\,\mathbb E[cc^\top]\big) = \mathrm{tr}\big(M(\Sigma_c +
\mu\mu^\top)\big) = \mathrm{tr}(M\Sigma_c) + \mu^\top M \mu$. Subtracting
$\mu^\top M \mu$ leaves $\mathrm{tr}(M\Sigma_c)$.
\end{proof}

\subsection{Equality-constrained QP: the Markowitz case}
\label{sec:qp-equality}

Markowitz mean-variance portfolio selection imposes at least a budget
constraint $\mathbf 1^\top z = 1$ \citep{markowitz1952portfolio}, so the relevant case is $Z = \{z : Az = b\}$ for a full-row-rank matrix $A \in \mathbb R^{m \times n}$, $m<n$, not the unconstrained case of
Corollary~\ref{cor:qp-unconstrained}. This case turns out to be as favorable as the unconstrained one, which the original manuscript did not notice: \emph{equality} constraints do not break the affine structure of $\pi^*$, because they enter the KKT system linearly and without case-based branching. Only \emph{inequality} constraints, treated in Section~\ref{sec:qp-inequality}, introduce genuine nonlinearity.

\begin{corollary}[Equality-constrained QP]
\label{cor:qp-equality}
For $Z = \{z \in \mathbb R^n : Az = b\}$ with $A$ full row rank and $Q
\succ 0$, the Lagrangian stationarity condition $c + \lambda Qz + A^\top \nu = 0$, together with $Az=b$, has a unique solution
\citep{nocedal2006numerical,boyd2004convex}
\[
\pi^*(c) = Mc + m_0, \qquad
M := -\tfrac1\lambda Q^{-1}\Big(I - A^\top(AQ^{-1}A^\top)^{-1}AQ^{-1}\Big),
\qquad
m_0 := Q^{-1}A^\top(AQ^{-1}A^\top)^{-1}b.
\]
This is affine in $c$ (setting $A=0$ recovers
Corollary~\ref{cor:qp-unconstrained}), so $R(c)=0$ for every distribution of $c$ by Theorem~\ref{thm:exact}(a), and
\[
\mathrm{Regret}(c) = -\mathrm{Cov}(c,\pi^*(c)) = -\mathrm{tr}(M\Sigma_c)
\ge 0
\]
by Lemma~\ref{lem:cov-linear} (nonnegativity follows from
Theorem~\ref{thm:jensen}, since $Regret \geq 0$ always). In particular,
\textbf{Markowitz's budget constraint does not by itself introduce any
residual}.
\end{corollary}

\begin{proof}
Form the Lagrangian $L(z,\nu) = c^\top z + \tfrac\lambda2 z^\top Qz +
\nu^\top(Az-b)$. Stationarity in $z$ gives $c + \lambda Qz + A^\top \nu =0$, i.e.\ $z = -\tfrac1\lambda Q^{-1}(c+A^\top\nu)$. Substituting into $Az=b$:
\[
-\tfrac1\lambda AQ^{-1}c - \tfrac1\lambda AQ^{-1}A^\top \nu = b
\;\Longrightarrow\;
\nu = -(AQ^{-1}A^\top)^{-1}\big(\lambda b + AQ^{-1}c\big),
\]
using that $AQ^{-1}A^\top \succ 0$ is invertible since $A$ has full row
rank and $Q \succ 0$. Substituting back,
\[
z^*(c) = -\tfrac1\lambda Q^{-1}c
+ Q^{-1}A^\top(AQ^{-1}A^\top)^{-1}b
+ \tfrac1\lambda Q^{-1}A^\top(AQ^{-1}A^\top)^{-1}AQ^{-1}c,
\]
which is $Mc+m_0$ as stated, after collecting the $c$-linear
terms into $M$.
\end{proof}

\subsection{Inequality-constrained QP: piecewise-affine structure}
\label{sec:qp-inequality}

\begin{lemma}[The constrained QP solution is a $Q$-weighted projection]
\label{lem:q-projection}
Let $z_{\mathrm{uc}}(c) := -\tfrac1\lambda Q^{-1}c$ denote the
unconstrained minimizer, and let $Z$ be any nonempty closed convex set.
Then
\[
\pi^*(c) = \arg\min_{z \in Z}\; (z - z_{\mathrm{uc}}(c))^\top Q\,(z - z_{\mathrm{uc}}(c))
\;=:\; \mathrm{Proj}^{\,Q}_Z\big(z_{\mathrm{uc}}(c)\big),
\]
the projection of $z_{\mathrm{uc}}(c)$ onto $Z$ in the $Q$-weighted inner product $\langle x,y\rangle_Q := x^\top Q y$. This coincides with the Euclidean projection $P_Z$ only when $Q$ is a multiple of the identity.
\end{lemma}

\begin{proof}
Complete the square using $Q$ symmetric and $z_{\mathrm{uc}}^\top Q =
-\tfrac1\lambda c^\top Q^{-1}Q = -\tfrac1\lambda c^\top$:
\begin{align*}
\tfrac\lambda2 (z-z_{\mathrm{uc}})^\top Q(z-z_{\mathrm{uc}})
&= \tfrac\lambda2 z^\top Qz - \lambda z_{\mathrm{uc}}^\top Qz + \tfrac\lambda2 z_{\mathrm{uc}}^\top Q z_{\mathrm{uc}} \\
&= \tfrac\lambda2 z^\top Qz + c^\top z + \tfrac\lambda2 z_{\mathrm{uc}}^\top Q z_{\mathrm{uc}}.
\end{align*}
The last term does not depend on $z$, so minimizing $c^\top z +
\tfrac\lambda2 z^\top Qz$ over $z \in Z$ is equivalent to minimizing
$(z-z_{\mathrm{uc}})^\top Q(z-z_{\mathrm{uc}})$ over $z \in Z$, which is, by definition, the $Q$-weighted projection of $z_{\mathrm{uc}}$ onto $Z$.
\end{proof}

\begin{proposition}[Piecewise-affine structure]
\label{prop:pa-qp}
Let $Z = \{z : Az=b,\ Gz \le h\}$ and suppose $Q \succ 0$. For a cost
vector $c$ at which the active-set solution satisfies the linear
independence constraint qualification (LICQ) and strict complementarity,
let $I(c) := \{i : G_i z^*(c) = h_i\}$ denote the active inequality
constraints. Then there is a neighborhood $U$ of $c$ on which $I(\cdot)$
is constant, and on $U$,
\[
\pi^*(c') = M_{I(c)}\,c' + m_{I(c)}, \qquad c' \in U,
\]
where $M_{I(c)}, m_{I(c)}$ are obtained by applying
Corollary~\ref{cor:qp-equality} to the equality-constrained QP with
constraint matrix $\big[\begin{smallmatrix}A\\ G_{I(c)}\end{smallmatrix}\big]$
in place of $A$. Consequently, $\pi^*$ is a \emph{piecewise-affine}
function of $c$ on $\mathrm{conv}(\mathcal C)$, affine on each region
$\mathcal C_I := \{c : I(c) = I\}$ and generally discontinuous across
region boundaries \citep{fiacco1983sensitivity,bank1982nonlinear,best2017quadratic}.
\end{proposition}

\begin{proof}[Proof sketch]
Under LICQ and strict complementarity at $c$, the active set is locally
stable under perturbation of $c$ (standard parametric-programming
stability result; see \citealp{fiacco1983sensitivity} Ch.~2 and
\citealp{bank1982nonlinear} Ch.~3, or the implicit-function-theorem
treatment in \citealp{bonnans2000perturbation}). On the resulting
neighborhood, the active inequality constraints $G_{I(c)}z=h_{I(c)}$ can be treated as equality constraints, reducing the problem to
Corollary~\ref{cor:qp-equality} with constraint matrix
$[A;G_{I(c)}]$. Piecewise affinity of $\pi^*$ over a polyhedral partition of $\mathrm{conv}(\mathcal C)$ into active set regions follows because there are finitely many possible active sets.
\end{proof}

\begin{corollary}[$R(c)$ is generically nonzero for inequality-constrained QPs]
\label{cor:qp-generic-nonzero}
By Proposition~\ref{prop:pa-qp}, $\pi^*$ is affine on each region
$\mathcal C_I$ but not, in general, globally affine on $\mathrm{conv}(\mathcal C)$. As in the LP case (Example~\ref{ex:lp-counterexample}, which is the special case $Q=0$ of this same piecewise-affine structure): if $\mathrm{conv}(\mathrm{supp}(c))$ lies within a single region $\mathcal C_I$, then $R(c)=0$, by Theorem~\ref{thm:exact}(a) applied to $\pi^*|_{\mathcal C_I}$; if $\mathbb E[c]=0$, then $R(c)=0$ trivially, by Theorem~\ref{thm:exact}(b); otherwise $R(c)$ is generically nonzero once $\mathrm{supp}(c)$ straddles more than one active-set region, and is bounded, whenever $Z$ is additionally compact (e.g.\ under a no-short-selling constraint $z \ge 0$ together with $\mathbf 1^\top z = 1$), by the distribution-free bound of Theorem~\ref{thm:bound}.
\end{corollary}

\begin{table}[h]
\centering
\begin{tabular}{lll}
\toprule
QP variant & $\pi^*(c)$ & $R(c)$ \\
\midrule
Unconstrained, $Q\succ0$ & $-\tfrac1\lambda Q^{-1}c$ (Cor.~\ref{cor:qp-unconstrained}) & $0$, \\
Equality-constrained (Markowitz budget) & $Mc+m_0$ (Cor.~\ref{cor:qp-equality}) & $0$, \\
Inequality-constrained (e.g.\ no short sales) & piecewise affine (Prop.~\ref{prop:pa-qp}) & generically $\ne 0$ (Cor.~\ref{cor:qp-generic-nonzero}) \\
\bottomrule
\end{tabular}
\caption{Summary: Residual for quadratic programs.}
\label{tab:qp}
\end{table}

\subsection{Real-Data Portfolio Validation}
\label{app:portfolio}
Figure~\ref{fig:portfolio} shows our metric's performance vis-a-vis a real-life financial portfolio.

We used monthly CRSP data for all U.S.\ listed securities from January 2015 through December 2024, removing observations with missing returns, returns outside $[-100\%,1000\%]$, non-positive
prices, micro-cap stocks under \$5M market capitalization, and stocks with fewer than 60 months of trading history, leaving 4{,}487 stocks across 119 monthly observations. Each month, we formed 100 portfolios of 50 stocks drawn with replacement, computed optimal Markowitz decisions, and forecast expected regret using the covariance formula; the following month's realized regret was then compared against this forecast. Figure~\ref{fig:portfolio} in the Appendix shows that, in aggregate, empirical regret approaches the theoretical prediction over time.

\begin{figure}
    \centering
    \includegraphics[width=0.5\linewidth]{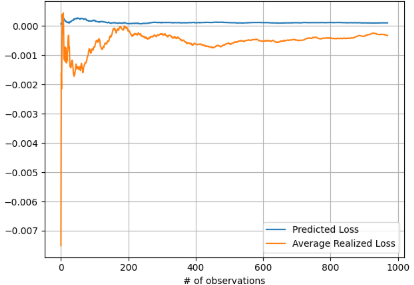}
    \caption{Cov-Predicted vs. Realized Regret in Real-data financial portfolios}
    \label{fig:portfolio}
\end{figure}

\begin{table}[h!]
\centering
\begin{tabular}{lll}
\toprule
Problem class & $R(c) = 0$? & Basis \\
\midrule
Unconstrained QP ($Q \succ 0$) & Exact, always & Thm.~\ref{thm:exact}(a), Cor.~\ref{cor:qp-unconstrained} \\
Any class, if $\mathbb E[c]=0$ & Exact, but distribution-specific & Thm.~\ref{thm:exact}(b) \\
LP, $\mathrm{supp}(c)$ within one normal cone & Exact & Thm.~\ref{thm:exact}(a) (degenerate case) \\
LP, general position & Generically $R(c) \ne 0$ & Example~\ref{ex:lp-counterexample} \\
Constrained QP, general & Generically $R(c) \ne 0$ & no affine $\pi^*$ \\
Integer / combinatorial $Z$ & Generically $R(c) \ne 0$, can be large & piecewise-constant $\pi^*$ \\
\bottomrule
\end{tabular}
\caption{Summary of What Is and Is Not Exact: Residual term $R(c)$. The distribution-free bound of Theorem~\ref{thm:bound} applies in every row regardless of whether $R(c)=0$.}
\label{tab:residual}
\end{table}

\section{Regret vs. Counterfactuals}

\begin{remark}[Regret as a counterfactual contrast]
\label{rmk:counterfactual}
$\mathrm{Regret}$ (Definition~\ref{def:regret-setup}) compares two
decision \emph{rules} applied to the same random cost. It is naturally
read as a counterfactual quantity in three related but distinct senses
that are worth separating explicitly, since the word "regret" may have different meaning across the literature streams considered in this paper.

\paragraph{(i) An identified potential-outcomes contrast.} In
the language of the Rubin causal model
\citep{rubin1974estimating,imbens2015causal}, define for each
realization $c$ two potential outcomes: $Y(1)(c) := c^\top\pi^*(c)$,
the cost incurred by the \emph{fully informed} decision rule, and
$Y(0)(c) := c^\top\pi^*(\mu)$, the cost incurred by the
\emph{certainty-equivalent} decision rule, both evaluated at the same realized $c$. Since $\pi^*(\mu)$ is a fixed vector, $\mathbb E[Y(0)] =
\mathbb E[c]^\top\pi^*(\mu) = \mu^\top\pi^*(\mu)$, so
\[
\mathbb E[Y(0)] - \mathbb E[Y(1)] = \mu^\top\pi^*(\mu) - \mathbb E[c^\top\pi^*(c)] = \mathrm{Regret}
\]
by Definition~\ref{def:regret-setup}: $\mathrm{Regret}$ is the
average treatment effect, on realized cost, of switching from the
certainty-equivalent policy to the fully informed one. What makes this
an unusually clean instance of a potential-outcomes contrast is that
\emph{both} $Y(0)(c)$ and $Y(1)(c)$ are computable for every $c$ once
$(Z,\pi^*)$ is specified. There is no missing-data problem, no
overlap or unconfoundedness assumption required, and no need for an
estimator of the treatment effect itself. This is counterfactual \emph{policy evaluation} within a fully specified optimization model 
rather than causal effect estimation from partially observed data. This
is also the sense in which classical information value theory
\citep{howard1966information} treats information as valuable: the value
of information is the expected improvement from the counterfactual world in which it is known.

\paragraph{(ii) Not the same object as counterfactual regret in game theory.}
The term "counterfactual regret," as used in the extensive-form-game
and online-learning literature following \citet{zinkevich2007regret}, denotes something structurally different: for a player $i$ at an information set $I$ in a repeated or extensive-form game, the counterfactual regret for an action $a$ is the difference in
counterfactual value (expected utility weighted by the probability of
reaching $I$ under the \emph{other} players' actual strategies, holding
player $i$'s own strategy fixed at reaching $I$) between playing $a$
and following the current strategy, accumulated over repeated play and
used algorithmically (via regret matching) to drive strategies toward a
Nash equilibrium. This is a multi-agent, sequential, per-information-set quantity used to \emph{compute} an equilibrium. $\mathrm{Regret}$ in this paper is a single-agent, single-shot, distribution-level quantity used to \emph{measure} the value of information in a fixed decision problem, with no adversarial players, no repeated play, and no algorithmic update rule attached to it. The two notions share only the root intuition common to every use of the word: an actual outcome compared against a counterfactual alternative. However, the two metrics are different and cannot be assumed to share results.

\paragraph{(iii) A counterfactual explanation of the certainty-equivalent decision.}
The introduction's original motivation for this line of work is that
decision-replication frameworks can help make algorithmic decisions more interpretable. In the terminology of counterfactual explanations in interpretable machine learning \citep{wachter2017counterfactual}, $\pi^*(c)$ is a counterfactual explanation of the certainty-equivalent decision $\pi^*(\mu)$: it answers "what decision would have been made had the true cost $c$ been known instead of only its mean $\mu$," and $\mathrm{Regret}$ quantifies, in expectation over $c$, how much that counterfactual decision would have been worth. The covariance decomposition of Proposition~\ref{prop:decomp} refines this further: $\mathrm{Cov}(c, \pi^*(c))$ measures how strongly the counterfactual decision rule's output actually moves with the realized cost, that is, how \emph{responsive} the fully informed policy is. At the same time, the residual $R(c)$ captures curvature in that responsiveness left over once the linear (affine) part is accounted for. However, nothing in Sections~\ref{sec:theory}--\ref{subsec:asymptotic} treats
"responsiveness" as a defined quantity.
\end{remark}

\subsection{Why the Distribution-Free Bound Loosens with Scale}
\label{app:bound-scaling}

We used the same monotone-lattice grid shortest-path family already introduced in Example~\ref{ex:shortest-path} (right/down moves only, the standard construction used throughout the decision-focused-learning literature, including the $5\times5$ benchmark of \citet{elmachtoub2022smart}), and scaled the grid from $5\times5$ up to $30\times30$. This specification matches and exceeds the $12\times12$-to-$32\times32$ range used in the Warcraft 
experiments of \citet{vlastelica2020differentiation}. In this family,
every feasible path uses $2(k-1)$ edges for a $k\times k$ grid, so
$D=\sqrt{2(k-1)}$ requires no enumeration or approximation to compute.
Edge costs were drawn i.i.d. log-normal. For each grid size, we
computed the true regret via dynamic programming over $6{,}000$ Monte
Carlo draws. The process was fast enough to solve directly rather than approximate even at $k=30$, since the monotone-lattice shortest path is a simple $O(k^2)$ grid recursion. We then compared it against Theorem~\ref{thm:bound}'s bound.
 
Both $D$ and $\sqrt{\mathrm{tr}(\Sigma_c)}$ grow with $k$ in
Table~\ref{tab:bound-scaling}'s grid family: $D = \Theta(\sqrt k)$, and
$\mathrm{tr}(\Sigma_c) = \Theta(k^2)$ since the number of edges grows
quadratically, so $\sqrt{\mathrm{tr}(\Sigma_c)} = \Theta(k)$, giving a
bound that grows like $\Theta(k^{1.5})$. The true regret is driven by
the covariance between costs and the specific path the optimizer
reroutes to under a cost perturbation, not by the total variance of
every edge in the grid, and it grows more slowly as a result. The gap is the information that the bound discards: it depends only on the
radius of $Z$ and the total cost variance, never on which particular
alternative paths or which specific extreme points of the path polytope a given cost perturbation is likely to make optimal. Closing this gap would require a bound that uses the probability mass split across $Z$'s normal fan instead of $D$ and $\mathrm{tr}(\Sigma_c)$ alone, as proposed in Section~\ref{sec:exp-bound-scaling}. Table~\ref{tab:bound-scaling} provides quantitative motivation for it rather than a speculative case.
 
A caveat specific to Warcraft-style benchmarks: the monotone-lattice
family tested in Table~\ref{tab:bound-scaling} has an exact, fixed $D=2(k-1)$, since every feasible path has the same length by construction. Instead, The Warcraft benchmark allows $8$-connected movement (diagonal steps), under which feasible paths of different lengths are extreme points of the same feasible region. Thus, $D=\sup_{z\in Z}\|z\|$ is no longer pinned down by the grid size alone; it is determined by the longest simple path the
formulation admits as feasible, which can exceed the length of any path
an optimizer would actually choose by a wide margin. The bound's
looseness on the literal Warcraft benchmark should be at least as severe as Table~\ref{tab:bound-scaling} shows, and plausibly worse for this reason.

\end{document}